\documentclass[11pt]{article}

\usepackage[utf8]{inputenc}
\usepackage[T1]{fontenc}
\usepackage[english]{babel}

\usepackage[pdftex]{graphicx}
\usepackage[usenames,dvipsnames]{xcolor}

\usepackage{booktabs}

\usepackage{tikz}
\usetikzlibrary{calc}

\usepackage{amsmath,amssymb,amsfonts,amsthm}
\usepackage{mathabx}
\usepackage[ruled,linesnumbered]{algorithm2e}

\usepackage{hyperref}
\usepackage{cleveref}

\usepackage{fullpage}
\usepackage{enumitem}

\newcommand{\eps}{\varepsilon}

\newcommand{\ceil}[1]{\left\lceil{#1}\right\rceil}

\newcommand{\abs}[1]{\left | #1 \right |}

\newcommand{\floor}[1]{\left \lfloor #1 \right \rfloor}
\newcommand{\D}{\mathcal{D}}
\renewcommand{\P}{\mathcal{P}}
\newcommand\drop[1]{}
\newcommand\E{{\textnormal{E}}}
\newcommand\U{{\mathcal U}}
\newcommand\C{{\mathcal C}}
\def\lcm{\operatorname{lcm}}
\renewcommand{\H}{\textnormal{H}}
\newcommand{\I}{\textnormal{I}}
\renewcommand{\subset}{\subseteq}
\newcommand{\Prp}[1]{\Pr\left [ #1 \right ]}
\newcommand{\Prpp}[2]{\Pr_{#1}\left [ #2 \right ]}

\newcommand{\Epp}[2]{\E_{#1}\left [ #2 \right ]}

\newcommand{\A}{\mathcal{A}}

\newtheorem{lemma}{Lemma}
\newtheorem{theorem}{Theorem}
\newtheorem{corollary}{Corollary}

\newtheorem{proposition}{Proposition}
\theoremstyle{definition}
\newtheorem{definition}{Definition}

\newcommand{\set}[1]{\left \{ #1 \right \}}

\usepackage{todonotes}

\usepackage{authblk}
\newcommand*\samethanks[1][\value{footnote}]{\footnotemark[#1]}

\title{The Entropy of Backwards Analysis}

\author{Mathias Bæk Tejs Knudsen\thanks{Research partly supported by
Advanced Grant DFF-0602-02499B from the Danish Council for Independent Research
under the Sapere Aude research career programme.} \thanks{Research partly
supported by the FNU project AlgoDisc - Discrete Mathematics, Algorithms, and
Data Structures.}}
\author{Mikkel Thorup\samethanks[1]}
\affil{University of Copenhagen\\\texttt{mathias@tejs.dk, mikkel2thorup@gmail.com}}
\date{}

\begin{document}

\maketitle

\begin{abstract}
Backwards analysis, first popularized by Seidel,  is often the simplest most elegant way of analyzing
a randomized algorithm. It applies to incremental algorithms where
elements are added incrementally, following some random permutation,
e.g., incremental Delauney triangulation of a pointset, where points 
are added one by one, and where we always maintain
the Delauney triangulation of the points added thus far. For backwards
analysis, we think of the permutation as generated backwards, implying
that the $i$th point in the permutation is picked uniformly 
at random from the $i$ points not picked yet in the backwards direction.
Backwards analysis has also been applied elegantly by Chan to the 
randomized linear time minimum spanning tree algorithm of Karger, Klein,
and Tarjan. 

The question considered in this paper is how much randomness we
need in order to trust the expected bounds obtained using backwards analysis,
exactly and approximately. For the exact case, it turns out that a
random permutation works if and only if it is minwise, that is, for any
given subset, each element has the same chance of being first.
Minwise permutations are known to have $\Theta(n)$ entropy, and
this is then also what we need for exact backwards analysis.

However, when it comes to approximation, the two concepts diverge
dramatically. To get backwards analysis to hold within a factor
$\alpha$, the random permutation needs entropy $\Omega(n/\alpha)$. This
contrasts with minwise permutations, where it is known that a $1+\eps$
approximation only needs $\Theta(\log (n/\eps))$ entropy. Our negative
result for backwards analysis essentially shows that it is as abstract
as any analysis based on full randomness.
\end{abstract}

\section{Introduction}
Randomization is a powerful tool in the construction of algorithms and
data structures, yielding algorithms that are often much simpler and
more efficient than their deterministic counterparts. However, much of
the analysis relies on the assumption of full access to randomness,
e.g., a sequence of $n$ i.i.d. random variables, a hash function
mapping keys from a large universe into independent hash values, or as
in this paper, a uniformly random permutation $\pi\in S_n$. However,
often $n$ is large, and then it is unrealistic to get access to just
$n$ independent random bits. A lot of work has been devoted to
removing the unrealistic assumption of full randomness, replacing it
with weaker notions of pseudo-randomness where the above large objects
are not uniformly random, but rather generated based on only few fully
random bits. We also want the generation to be efficient, e.g., for a
pseudo-random hash function, we want to quickly compute the
hash value of any given key. While the pseudo-random objects are not
uniformly random, we do want them to satisfy probabilistic properties
that can be proved to suffice in many algorithmic contexts. A classic
example is Wegman and Carter's {\em $k$-independent hash functions\/}
\cite{wegman81kwise} with the property that any $k$ given keys get
independent random hash values. We can get such a hash function over
the prime field $\mathbb Z_p$ using a random degree $k-1$ polynomial,
instantiated with $k$ independent random coefficients from $\mathbb
Z_p$.  The important point here is that if we have an application of
hash functions where the analysis only assumes uses $k$-independence,
then we know that it can be implemented using any $k$-independent hash
function. Another example more relevant to this paper is minwise
permuations of Broder et al. \cite{broder98minwise}. A random permutation $\pi\in S_n$
is {\em (approximately) minwise\/} if for any given subset $X\subseteq
[n] =\{1,\ldots,n\}$, each $x\in X$ has (approximately) the same
chance of being first from $X$ in $\pi$. If a minwise permutation is
used to pick the pivots in Quicksort \cite{hoare62quicksort}, then we
know that Quicksort will have exactly the same expected number of
comparisons as if a fully random permutation was used---two keys $x$
and $y$ get compared if and only if one of them is picked before any
key between them. Exact minwise permutations are unrealistic to
implement in the sense that they have entropy at least $\Theta(n)$ \cite{broder98minwise}, which is only
slightly less than a fully random permutation that has entropy $\ln
n!\approx n\ln n$.  However, if we are satisfied that each $x\in X$ is
first with probability at most $(1+\eps)/|X|$, then the required
entropy is only $\Theta(\log(n/\eps))$ \cite{broder98minwise}, and then the expected number
of comparisons in Quicksort is only increased by a factor $1+\eps$.

Randomized algorithms have now been studied for more than 50 years
\cite{hoare62quicksort,rabin76random}. We have a vast experience for
how randomized algorithms can be understood and
analyzed. Probabilistic properties have emerged that have proven
useful in many different contexts, and we want to know if these
properties can be supported with realistic implementations. Generally,
we want the simplest most elegant analysis of a randomized algorithm,
but if it relies on unrealistic probabilistic properties, then we have
to also look for a possibly more complicated analysis relying on more
realistic properties, or worse, give up the algorithm because we do
not know how to implement the randomness.

\subsection{Backwards analysis}
In this paper we focus on, so-called, backwards analysis, which was
first popularized by Seidel \cite{seidel93backwards}. Using backwards
analysis he obtained an extremely simple and elegant analysis for many
randomized algorithms assuming full randomness. The technique is so
appealing that it is now taught in text books
\cite{marks-cgbook,motwani95book}.

Backwards analysis applies to incremental algorithms where elements from some
set $X=\{X_1,...,X_n\}$ are added one by one in a random order,
starting from the empty set.  Let $\pi$ be the corresponding
permutation over $[n]=\{1,\ldots,n\}$, that is, $X_{\pi(i)}$ is the
$i$th element added to the set. We note that backwards
analysis has also been applied to the case where elements are added
in batches \cite{chan1998backwards} rather than one by one. 
We shall return to this in Section~\ref{sec:MST}, and for now just
focus on permutations.

We also have a cost function that with any $Y\subseteq X$ and $x\in Y$
associates a cost $c(x,Y)$ of adding $x$ to $Y\setminus\{x\}$.  A
classic example from \cite{seidel93backwards} is incremental Delauney
triangulation of a point set where $c(x,Y)$ is the degree of $x$ in
the Delauney triangulation of $Y$. For a given permutation $\pi$, the
total cost is $\sum_{i=1}^n c(X_{\pi(i)},X_{\pi([i])})$. Here
$[i]=\{1,\ldots,i\}$, $\pi(I)=\{\pi(j)\mid j\in I\}$, and $X_J=\{X_j
\mid j\in J\}$, so $X_{\pi([i])}=\{X_{\pi(1)},\ldots,X_{\pi(i)}\}$.
Now, what is the expected total cost if the permutation $\pi$ is
fully-random?

In backwards analysis, we think of the permutation $\pi$ as generated
{\em backwards} starting from $\pi(n)$. Then {\em $\pi(i)$ is uniformly
distributed in $\pi([i])=\{\pi(1),\ldots,\pi(i)\}=\{1,\ldots,n\}\setminus\{\pi(i+1),\ldots,\pi(n)\}$}. The important point here is that
when $\pi(i)$ is generated, we only know the set $\pi([i])$ via the
values of $\pi(i+1),\ldots,\pi(n)$---we do not now the individual values of
$\pi(1),\ldots,\pi(i)$.
For Delauney triangulation this means that the $i$th point
$X_{\pi(i)}$ is uniformly distributed amongst the points in the given
set $X_{\pi([i])}$. The expected cost is
therefore the average degree of a point in the (assumed unique) Delauney triangulation of
$X_{\pi([i])}$. Since the Delauney triangulation is a planar graph,
by Euler's formula, the average degree is less than $6$. The expected
cost of adding all $n$ points is thus at most $6n$.

In a more general use of backwards analysis, suppose for
$i=1,\ldots,n$, we have determined $c_i$ such that for any given
subset $Y\subseteq X$ of size $i$, if $x$ is uniformly distributed in
$Y$, then $\E[c(x,Y)]\leq c_i$. Then, with a fully random permutation, the
expected total cost is bounded by $\sum_{i=1}^n c_i$.  For Delauney
triangulation, we had $c_i=6$ for all $i$. However, there
are also applications where the $c_i$ are different, e.g., in the analysis of
quicksort from \cite{seidel93backwards}, we get that the expected
number of comparisons with the $i$th pivot is bounded by $c_i\leq 2n/i$, 
hence that the expected total number of comparisons is bounded by $2n H_n$.

Our question in this paper is how much
randomness we need to create a distribution on permutations so that
the expected bounds derived from backwards analysis apply exactly or
approximately. Based on the many positive findings surveyed in \cite{thorup15tabulation}, we were originally hoping for a good solution using
$O(n^\eps)$ random bits for some small $\eps>0$.

We will now define more precisely the probabilistic properties needed
for backwards analysis. To simplify notation, we identify our
set $X$ of elements with $[n]=\{1,\ldots,n\}$, that is, $x_i=i$ for
$i=1,\ldots,n$. A {\em cost function over $[n]$} assigns a cost
$c(x,Y)$ to any subset $Y\subseteq [n]$ and $x\in Y$. The total cost
of a permutation $\pi\in S_n$ is then
\[c(\pi)=\sum_{i=1}^n c(\pi(i),\pi([i])).\]
For our most positive results, we would like 
our distribution to be backwards $\alpha$-uniform, defined as follows.
\begin{definition}\label{def:uniform}
A distribution $\D$ over $S_n$ is {\em backwards $\alpha$-uniform\/} if for any
set $Y\subset [n]$ with positive support, that is, 
$\Pr_{\pi\sim D}[\pi([|Y|])=Y]>0$, and any $x\in Y$, we have
that 
\[\Pr_{\pi\sim D}[\pi(|Y|)=x\mid \pi([|Y|])=Y]\leq \alpha/|Y|.\]
If $\D$ is backwards $1$-uniform, we will also say that it is {\em exact
backwards uniform}.
\end{definition}
Let $c$ be {\em any\/} cost function associating a cost $c(x,Y)$ with
any $Y\subset [n]$ and $x\in Y$.  For every $i\in [n]$, let $c_i$ be
such that for every size $i$ subset $Y\subset [n]$, we have $\E_{x
  \sim \U(Y)}[c(x,Y)]\leq c_i$. If $\pi$ is drawn from an
$\alpha$-uniform distribution $\D$, then it is easily seen that the
total expected cost is $\E_{\pi\sim D}[c(\pi)]\leq \alpha\sum_{i=1}^n
c_i$, which is $\alpha$ times bigger than the bound $\sum_{i=1}^n c_i$
we get assuming full randomness like in backwards analysis.  Demanding
$\alpha$-uniformity is, however, very conservative. All we really need
from the distribution $\D$ is that $\E_{\pi\sim D}[c(\pi)]\leq
\alpha\sum_{i=1}^n c_i$ for any given $c$ and $c_i$s satisfying the
conditions.  We are going to prove lower bounds for this condition
even for the restricted case where all $c_i$ are the same and
normalized to $1$. We refer to this as being {\em backwards
  $\alpha$-efficient\/} as defined below.
\begin{definition}\label{def:efficient}
A distribution $\D$ over $S_n$ is {\em backwards
$\alpha$-efficient\/} if the following condition is satisfied.
Let $c$ be {\em any\/} cost function that with any $Y\subset [n]$ and $x\in Y$
associates a cost $c(x,Y)$ such that for every subset
$Y\subset [n]$, we have $\E_{x \sim \U(Y)}[c(x,Y)]\leq 1$.
Then 
\begin{align}
	\notag
	\E_{\pi\sim D}[c(\pi)]\leq \alpha n.
\end{align}
If $\D$ is backwards $1$-efficient, we will also say that it is {\em exact
backwards efficient}.
\end{definition}
If the distribution $\D$ is backwards $\alpha$-uniform, it is clearly
also backwards $\alpha$-efficient. Note in the definition of
$\alpha$-efficient that the distribution $\D$ over $S_n$ is given
first, and that it has to work for any later cost function. The vision is
that we want a ``standard library'' implementation of the distribution
that gives the right expectation for any future cost function.
Unfortunately, our main result is negative: as described below, if $\D$ is
$\alpha$-efficient, then we need $\Omega(n/\alpha)$ random bits to
define a random permutation $\pi$ from $\D$. 

\subsection{Results}
We now list our results in more detail.  First, concerning
exact backwards analysis, we show that a
backwards 1-uniform distribution can be implemented as a uniform
distribution over a family of $\exp(O(n))$ permutations. We have
a matching lower bound showing that even if we allow non-uniform
distributions and even if we are satisfied with 1-efficiency, then the
family has to be of size $\exp(\Omega(n))$.

Our more interesting results are for $\alpha$-approximations for
$\alpha>1$.  On the positive side we present a backwards
$\alpha$-uniform distribution which is a uniform distribution over a
family of size $\exp(O(n(\log\alpha)^2/ \alpha))$. On the negative
side, we show that even if we allow non-uniform distributions and even
if we are satisfied with $\alpha$-efficiency, then the family has to
be of size $\exp(\Omega(n/\alpha))$. More precisely, we show that if a
random permutation is $\alpha$-efficient, then its entropy is
$\Omega(n/\alpha)$. 

The entropy lower bound on the family size needed for
backwards $\alpha$-efficient permutations is our main 
result. For a given distribution $\D$, we define a simple
adversarial cost function $c^\D$ as follows. Consider a set $Y\subseteq [n]$.
If $Y\subseteq [n]$ does not have positive support in $\D$, that is,
$\Pr_{\pi\sim D}[\pi([|Y|])=Y]=0$, then we just set $c(y,Y)=1/|Y|$
for all $y\in Y$. If $Y$ has positive support, we let $x\in Y$
maximize the chance that it is last in $Y$, that is, $x$ maximizes
$\Pr_{\pi\sim D}[\pi(|Y|)=x\mid \pi([|Y|])=Y]$.  We then set $c^\D(x,Y)=|Y|$
and $c^\D(y,Y)=0$ for $y\in Y\setminus\{x\}$. Then 
$\E_{x \sim \U(Y)}[c^D(x,Y)]\leq 1$ for all $Y$ as required.
If $\D$ is backwards $\alpha$-efficient, then the expected total
cost should be $\E_{\pi\sim D}[c^\D(\pi)]\leq \alpha n$. However,
here we prove that $\E_{\pi\sim D}[c^\D(\pi)]\leq \alpha n$ implies
that the entropy of $\D$ is $\Omega(n/\alpha)$.

\subsection{Comparison with minwise permutation}
Our results are contrasted with corresponding results for minwise
permutations \cite{broder98minwise}. An {\em $\eps$-minwise\/} distribution
$\D$ over permutations from $S_n$ is one such that 
for any given $Y\in [n]$ and any $x\in Y$, we have
\[\Pr_{\pi\sim D}[\pi(x)=\min \pi(Y)]=(1\pm\eps)/|Y|.\]
If $\eps=0$, we say $\D$ is exact minwise. In our condition of
backwards $\alpha$-uniform permutations, we only have to consider sets
$Y$ with positive support, and then we want
\[\Pr_{\pi\sim D}[\pi(x)=\max \pi(Y)\mid \pi([|Y|])=Y]\leq \alpha/|Y|.\]
Above, the difference between using minimum and maximum is inconsequential.
The important difference is the conditioning $\pi([|Y|])=Y$.

In this paper, we show that for a random permutation $\D$,
the following properties are equivalent:
\begin{itemize}
\item $\D$ is exact minwise.
\item $\D$ is exact backwards uniform.
\item $\D$ is exact backwards efficient.
\end{itemize}
From \cite{broder98minwise} we know that the family size needed
for an exact minwise distribution is $\exp(\Theta(n))$, and hence this
is the size we need for exact backwards analysis.

However, when it comes to approximation, it turns out that there is a
huge difference between minwise and backwards permutations.  
For $\eps$-minwise distribution, we have very positive results. In
\cite{broder98minwise} it is proved that for $\eps$-minwise permutations, it suffices
with a uniform distribution on a family of size $O(n^2/\eps^2)$. 
In fact, more constructively, Indyk \cite{indyk01minwise} has show that it 
suffices to use $O(\log (1/\eps))$-independence.
These positive results for $\eps$-minwise stand in sharp contrast to our lower bound of $\exp(\Omega(n/\alpha))$ on the family size needed for a distribution over $S_n$ to be backwards $\alpha$-efficient.

Viewed in terms of random bits, for constant approximations, the above
results say that we need only need a logarithmic bits for approximate
minwise permutations whereas we for backwards analysis need a linear
number of bits.

\subsection{Chan's backwards analysis of the Karger-Klein-Tarjan's MST}\label{sec:MST}
Chan \cite{chan1998backwards} has made an interesting application of
backwards analysis for the case where elements are not added one by
one as in, but rather in a batches. The point is that large batches
may be handled more efficiently than elements processed one element at
the time. The concrete case is the randomized minimum spanning tree
algorithm Karger, Klein, and Tarjan \cite{KKT95:mst} which runs in linear
expected time. We assume that all edge weights are unique.

The important step in the algorithm is that we have a graph $G=(V,E)$
from which we sample a random subset $S\subset E$ of the
edges. Recursively, we compute a minimum spanning forest $F$ from of
$(V,S)$.  Finally, we remove all edges $(v,w)$ from $G$ that are {\em
  $F$-heavy\/} in the sense $v$ and $w$ are connected in $F$ by a path
with no edge heavier than $(v,w)$. In particular, this includes all
edges from $S\setminus F$.  The question is, how many edges do we
expect to remain in $G$?

Let us say that we start with $n$ nodes and $m$ edges in $G$. As in
\cite{chan1998backwards}, we simplify the discussion a bit by saying
that $S$ is sampled to have a specific size $pm$, and we assume
that $S$ is sampled uniformly. We claim that the expected number of edges
remaining in $G$ after the removal of $F$-heavy edges is less than $n/p$.

The proof is using the idea from backwards analysis. Consider adding
uniformly random edge $e\in E\setminus S$ to $S$. Then $e$ is uniformly
distributed in $S'=S\cup\{e\}$ just like $\pi(i)$ was uniformly
distributed in $\pi([i])$.

We now make the combinatorial observation that $e$ is $F$-light if and only
if $e$ is in the minimum spanning forest $F'$ of $(V,S')$. There are
at most $n-1$ edges in $F$, and $e$ is uniform in $S'$, so we conclude that
$e$ is $F$-light with probability $(n-1)/(pm+1)<n/(pm)$. Thus we expect
$(m-pm)n/(pm)=n/p-n$ edges to remain from $E\setminus S$. Including
the at most $n-1$ edges from $F$. It follows that the total expected
number of edges remaining is less than $n/p$, as desired.

The question now is if we can support this kind of backwards analysis
based generating a set $S$ of some prescribed size $pm$, followed by
an element $e$. As in the previous backwards analysis, we have
a cost function $c(S',e)$ of $e$ being added last to $S'=S\cup\{e\}$. The
cost is $1$ if $e$ is light. We argued for any set $S'$ that 
$\E_{e \sim \U(S')}[c(e,S')]<n/(pm)$, and based on that, we which to
conclude that $\E[\sum_{e\in E\setminus S} c(e,S\cup\{e\})]\leq |E\setminus S|n/(pm)$.

Stepping back, normalizing, and allowing for an approximation factor $\alpha$,
 what we would like in general
for this kind of analysis is a distribution 
 $\D$ over subsets $S$ of $E$ of size $pm$
such that if $c$ be {\em any\/} cost function such that
for any set $E'$ of size $pm+1$, we have
$\E_{e \sim \U(S')}[c(e,S')]\leq 1$, then
we conclude that 
\[\E_{S\sim \D}[\sum_{e\in E\setminus S} c(e,S\cup\{e\})]\leq 
\alpha |E\setminus S|.\]
Unfortunately, we will show that if $p$ is a constant from $(0,1)$, then
this implies that $\D$ has entropy $\Omega(m/\alpha \log(1+\alpha))$
unless the approximation factor is very large, i.e. $\alpha = \Omega(m)$.

Chan's MST analysis is an example illustrating how backwards analysis
can be applied when elements are added in random batches, and where
random elements from one batch are only considered against elements in
from preceding batches.  In the MST, we have only two batches: (1) the
initial samples set $S$, and (2) the remaining edges $E\setminus S$.
This is the extreme opposite of the backwards permutations, where
elements were added randomly in batches of size 1. Our entropy lower
bounds for both of these extreme cases suggests a general hardness.

\subsection{Reducing the randomness}
Getting a linear number of random bits in order to trust the bounds
from backwards analysis within a constant factor is often
prohibitive. It is almost as bad as using $\Theta(n\log n)$ bits to
get a fully-random permutation. 

For a concrete algorithm like the incremental Delauney triangulation,
one can still hope for a different analysis based on weaker
assumptions on the distributions of the random permutation
$\pi$. Indeed, Mulmuley \cite{mulmuley96notbackwards} has shown that
the bounds mentioned by Seidel \cite{seidel93backwards} can be
achieved within a constant factor using only $O(\log n)$
bits. Mulmuley's idea is very simple.  For the Delauney triangulation
he points out that a triangle $uvw$ will appear if and only if there
is no point $z$ inside the triangle appearing before the last of $u$,
$v$, and $w$ in the permutation, and the time used is determined by
the total number of triangles appearing. Mulmuley then showed that
generating the permutation using an $11$-independent function guarantees
that the probability that any given triangle appears is at most a
constant fraction larger than it would be, had the permutation been
generated using full randomness. It follows that Delauney
triangulation has asymptotically the same complexity with $11$-independence as
with full randomness, and full randomness was analysized with
backwards analysis. However, Mulmuley's idea does not apply to the
randomized MST algorithm from Section \ref{sec:MST}. The point is that
the obstruction for $e=(v,w)$ being $F$-light requires a whole path of
light edges from $v$ to $w$ to be sampled for $S$. Indeed it is an open
problem if the sample $S$ from the randomized MST algorithm can be 
generated with $o(n)$ bits. Pettie et al. \cite{PR08mst} have shown
that the MST algorithm can be changed to one that needs only $O(\log n)$
bits, but it would have been much more attractive if an
efficient implementation followed directly from the backwards analysis.

If we know in advance that we are only going to study a fixed limited
family of cost functions $\C_n$, then we can always use the standard
trick of first generating a limited number of fully-random
permutations $\Pi_n\subseteq S_n$, that we are free to use in our
permutation generator. Then, to handle a concrete cost function $c\in
\C_n$, we just have to pick a random $\pi\in\Pi_n$ using $\lg_2
|\Pi_n|$ bits.  This model prevents an adversary from learning $\Pi_n$
and create the adversarial cost function $c^{\,\U(\Pi_n)}$ discussed
earlier. Our goal in this paper is to have a public permutation
generator that works for all possible future cost functions, and this
is when a linear number of random bits are needed for backwards
analysis to be trusted. Such tricks can be applied whenever we have an
analysis assuming full randomness, and indeed the basic message of
this paper is that backwards analysis is almost as abstract as any
analysis based on full randomness.

\section{Preliminaries}

We define $[n] = \set{1,2,\ldots,n}$ and let $S_n$ denote the set of
permutations of $[n]$. For a graph $G$ we let $V(G)$ and $E(G)$ denote the set
of nodes and edges of $G$, respectively. For a set $S$ we let $\P(S)$ denote
the power set of $S$, i.e. the set containing all subsets of $S$.

For a random variable $X$ and a distribution $\D$ we write $X \sim \D$ to
mean that $X$ is drawn from $\D$. Given $X$ drawn from $\D$ the self information
$\I(X)$ of $X$ is defined as $\I(X) = -\ln \left ( \Pr_{X' \sim D} (X=X') \right )$,
where $X$ and $X'$ are i.i.d., and the entropy of $X$ is $\H(X) = \E(\I(X))$.
When $X$ is drawn from $\D$ the entropy of $\D$ is defined by $\H(\D) = \H(X)$.
\section{Transition Graphs of Distributions}

Given a distribution $\D$ on $S_n$ we associate a \emph{transition graph} $G$ defined
in the following manner. The nodes of $G$ are $\P([n])$, i.e., the set of all subsets
of $[n]$. For every node $S \in V(G)$
the weight of $S$ is $w_G(S) = \Prpp{\pi \sim \D}{\pi\left([\abs{S}]\right)=S}$.
For each non-empty set $S \subset [n]$ such that $w_G(S) > 0$ and $s \in S$
there is and edge from $S$ to $S \setminus \set{s}$ with weight
$w_G(S, S \setminus \set{s})$ defined by
\begin{align}
	\notag
	w_G(S, S \setminus \set{s})
	=
	\Prp{
		\pi\left(\abs{S}\right) = s
		\mid 
		\pi\left(\left[\abs{S}\right]\right) = S
	}
	\, .
\end{align}

Each permutation $\pi$ from $\D$ defines, in reverse order,
a walk from $[n]$ to $\emptyset$ where each edge drops an element from the current
set, that is, the walk $(\pi([n]), \pi([n-1]), \ldots, \pi([1]), \emptyset)$.

When it is clear from the context which graph $G$ we are working with, we write
$w$ instead of $w_G$.

\begin{definition}
	\label{def:transitionGraph}
	A \emph{transition graph} $G$ is a graph for which there exists a distribution
	$\D$ such that $G$ is the transition graph for $\D$. If $G$ is the transition
	graph for $\D$, we write $\D \to G$.
\end{definition}

The reason we introduce the transition graph is that
it captures the backwards performance of the distribution. The
distribution $\D$ is $\alpha$-uniform if and only if, for each node
$S$ with positive weight, every outgoing edges $(S,S\setminus\set{s})$
has weight at most $\alpha/|S|$. Likewise, for any given cost
function $c$, the expected cost with $\D$ can be computed from $G$ as
$\sum_{S,S\setminus\set{s}}(c(s,S)w(S)w(S,S\setminus\set{s})$.
The $\alpha$-efficiency of $\D$ is thus also determined from $G$.

Let $G^U_n$ be the transition graph of the uniform distribution over $S_n$.
Because a transition graph captures backwards performance,
we get that any permutation distribution $\D$ with transition graph
$G^U_n$ must be both backwards $1$-uniform and $1$-efficient. In fact,
we will also argue the converse: if $\D$ is backwards $1$-uniform or $1$-efficient,
then $\D\to G^U_n$. We will argue that this is also equivalent to
$\D$ being minwise or maxwise, that is, we will prove Proposition \ref{prop:minwiseEqualsBackwards}.

\begin{proposition}
	\label{prop:minwiseEqualsBackwards}
	Let $\D$ be a distribution and $G^U_n$ the transition graph of $\U(S_n)$. Then
	the following five statements are equivalent:
	(i) $\D$ is backwards $1$-uniform.
	(ii) $\D$ is (restricted) backwards $1$-efficient.
	(iii) $\D$ is exact minwise.
	(iv) $\D$ is exact maxwise.
	(v) $\D \to G^U_n$.
\end{proposition}
\begin{proof}
Let $G$ be the distribution graph of $\D$.
We have already seen that whether $\D$ is backwards $1$-uniform and whether $\D$
is backwards $1$-efficient, are properties of $G$. Since
$\U(S_n)$ is backwards $1$-uniform and $1$-efficient this implies that (v) $\Rightarrow$
(i), (ii). Let $\pi \sim \D$. Then $\D$ is exact minwise if and only if $\Prp{\min \pi(S) = \pi(s)} = \frac{1}{\abs{S}}$
for every $s \in S \subset [n]$. We can write $\Prp{\min \pi(S) = \pi(s)}$ in
terms of the weights of $G$ in the following manner
\begin{align}
	\label{eq:minwiseByWeights}
	\Prp{\min \pi(S) = \pi(s)}
	=
	\sum_{T \subset [n], T \cap S = \set{s}}
		w_G(T) w_G(T,T\setminus\set{s})
	\, .
\end{align}
In the same manner we note that
\begin{align}
	\label{eq:maxwiseByWeights}
	\Prp{\max \pi(S) = \pi(s)}
	=
	\sum_{T \subset [n], S \subset T}
		w_G(T) w_G(T,T\setminus\set{s})
	\, .
\end{align}
Equations \eqref{eq:minwiseByWeights} and \eqref{eq:maxwiseByWeights} imply that (iii)
and (iv) follows from (v).

We have proved that if $G = G^U_n$ then all of (i)-(iv) hold. To prove that all the
statements are equivalent we will prove that if $G \neq G^U_n$ then none of (i)-(iv) hold.
So assume that $G$ is not the distribution graph for $\U(S_n)$ and let $S \subset [n]$ 
and $s \in S$ be such that $w_G(S, S \setminus \set{s}) \neq w_{G^U_n}(S, S \setminus \set{s}) = \frac{1}{\abs{S}}$.
Such a $S$ much exists since the edge weights determine the node weights. 
We furthermore choose $S$ such that $\abs{S}$ is maximal under this constraint.
We also note that we may choose $s$ such that $w_G(S, S \setminus \set{s}) > \frac{1}{\abs{S}}$.
We observe that $w_G(T) = w_{G^U_n}(T)$ for
every set $T$ with $\abs{T} \ge \abs{S}$. Here we use again, that the edge weights determine the
node weights: $w_G(T)$ can be calculated by the weights
of the edges on paths from $[n]$ to $T$, and these weights are equal in $G$ and $G^U_n$.

Since $w_G(S,S \setminus \set{s}) > \frac{1}{\abs{S}}$ it is clear that $G$ is not
backwards $1$-uniform, and therefore (i) does not hold. Let $c(t,T)$ be the cost function
that equals $1$ when $T \neq S$, and for $T=S$ we let $c(t,T) = \abs{S}$ if $t = s$ and
$c(t,T) = 0$ for $t \neq s$. Then we have that
\begin{align}
	\notag
	\Epp{\pi \sim \D}{c(\pi)}
	-
	\Epp{\pi \sim \U(S_n)}{c(\pi)}
	=
	w_G(S) w_G(S,S \setminus \set{s}) \abs{S}
	-
	w_{G^U_n}(S)
	> 
	0
	\, ,
\end{align}
and therefore $\D$ is not backwards $1$-efficient, and (ii) does not hold.
Now consider the calculation of $\Prpp{\pi \sim \D}{\max \pi(S) = \pi(s)}$
in \eqref{eq:maxwiseByWeights}. All terms of the sum are the same if
we replace $G$ by $G^U_n$, except for the term $w_G(S) w_G(S,S\setminus\set{s})$.
Since $w_G(S) = w_{G^U_n}(S) > 0$ and $w_G(S,S\setminus\set{s}) > w_{G^U_n}(S,S\setminus\set{s})$
we conclude that the the probability that $\max \pi(S) = \pi(s)$ is larger when 
$\pi \sim \D$ that when $\pi \sim \U(S_n)$ and hence $\D$ is non exact maxwise and
(iv) does not hold.

We have thus far proved that (i), (ii), (iv), (v) are all equivalent and
they imply (iii) via (v). We have proved maxwise $\Rightarrow$ minwise,
so by symmetry minwise $\Rightarrow$ maxwise as well, so (iii) and (iv) are equivalent.
So all of the five conditions are equivalent, which finishes the proof.
\end{proof}

The minimal size of a set $X \subset S_n$ such that $\U(X)$ is an exact
minwise distribution was studied in
\cite{broder98minwise}, and by 
Proposition \ref{prop:minwiseEqualsBackwards}, this is also the
bound for backwards 1-uniform and 1-effient. In \cite{broder98minwise},
the minimal size was proved to be between $\lcm(1,2,\ldots,n) = e^{n-o(n)}$ and $4^n$.
In Corollary \ref{cor:optimalMinWiseUniform} we prove that the minimal size is
precisely $\lcm(1,2,\ldots,n)$.

Now let $G$ be a transition graph. We associate a \emph{memoryless
distribution} of permutations $\D$ to $G$.
The permutations of $\D$
are constructed in the following manner.  Let $[n] = A_n,A_{n-1},
\ldots, A_0$ be a random walk of length $n$ in $G$ starting in the
node $[n]$. When the walk is at the node $A$ the next node on the walk
is chosen independently from the previous choices and randomly from
the outgoing neighbours of $A$, such that $B$ is chosen with
probability $w(A,B)$ if there is an edge from $A$ to $B$.  In this way
the walk $A_n,A_{n-1},\ldots,A_0$ is a path in $G$ such that
$\abs{A_i} = i$ for $i=0,1,2,\ldots,n$. Let $\pi \in S_n$ be the
permutation such that $\set{\pi(i)} = A_i \setminus A_{i-1}$ for each
$i \in [n]$.  
The distribution of $\pi$ generated in the above manner is called the memoryless distribution
of $G$.

\begin{definition}
	\label{def:memorylessDistribution}
	Given a transition graph $G$ the \emph{memoryless distribution} $\D$ of $G$
	is the distribution on $S_n$ described above. If $\D$ is the memoryless
	distribution of $G$ we write $G \to \D$.
\end{definition}

We remark that for a transition graph there may be multiple distributions
$\D,\D'$ such that $\D \to G$ and $\D' \to G$. However, for each distribution
$\D$ there is only one transition graph $G$ such that $\D \to G$. If $G \to \D$ then
$\D \to G$, but $\D \to G$ does not imply that $G \to \D$.

We also note that if $\D \to G \to \D'$ then $\D$ and $\D'$ might be different,
but $G \to \D \to G'$ implies that $G = G'$. 

We remark, without proof, that among the distributions $\D$ with
$\D \to G$, the memoryless distribution of $G$ is the one with the maximal
amount of entropy.

\section{Lower Bound}

The main theorem of this paper, Theorem \ref{thm:lowerbound}, shows that if a distribution
$\D$ is $\alpha$-efficient then the entropy of $\D$ is $\Omega\left(\frac{n}{\alpha}\right)$.
The main idea is as follows. Let $G$ be the transition graph of $\D$ and 
$\D'$ the memoryless distribution of $G$ such that $\D \to G \to \D'$. We then
consider the $(q+1)$-tuple $\ell(\pi) = (\pi([p]), \pi([p+1]), \ldots, \pi([p+q]))$,
where $\pi$ is drawn from the memoryless distribution $\D'$, and where $p,q$ are suitably chosen depending on
$\D$.
We then proceed to give a lower bound on $\H(\ell(\pi))$ in terms of $n$ and $\alpha$,
and give a upper bound on $\H(\ell(\pi))$ in terms of $\H(\D)$.
The key idea is that $\H(\D)$ can be used to give an upper bound on $\H(\ell(\pi))$
even though $\pi$ is drawn from $\D'$ and not from $\D$. This comes from the fact that
$\ell(\pi)$ can be described by $\pi([p])$ and $\pi([p+q])$ along with a permutation of
$\pi([p+q]) \setminus \pi([p])$. Each of the sets $\pi([p])$ and $\pi([p+q])$ have the
same distribution whether $\pi$ is chosen from $\D$ or $\D'$ since this
is determined by the weights of the corresponding nodes in $G$, and so the entropy of
each set is at most $\H(\D)$. Since there are at most $q!$ permutations of $q$ elements
it follows that $\H(\ell(\pi)) \le 2\H(\D) + \ln(q!)$.
In order to give a lower bound on the entropy of $\ell(\pi)$
we use that the distribution
of $\pi([p+i])$ conditioned on $(\pi([p+i+1]), \pi([p+i+2]), \ldots, \pi([p+q]))$
is the same as the distribution of $\pi([p+i])$ conditioned on $\pi([p+i+1])$,
because the distribution is memoryless.
Informally, this means that when we discover $\pi([p+i])$ after having discovered
$(\pi([p+i+1]), \pi([p+i+2]), \ldots, \pi([p+q]))$ there is still a large amount
of entropy in $\pi([p+i])$ - the exact amount is determined by the weights of the
outgoing edges of $\pi([p+i+1])$ in $G$. If $\pi$ had been drawn from $\D$ instead then
this would not necessarily be the case and the proof would break down.
The technical details are given below.

\begin{theorem}
	\label{thm:lowerbound}
	Let $\D$ be a distribution over $S_n$. If $\D$ is backwards
	$\alpha$-efficient, then $\H(\D) \ge \floor{\frac{n}{48\alpha}}$.
\end{theorem}
We have not tried to optimize the constant in Theorem \ref{thm:lowerbound}.
\begin{proof}
Let $G$ be the transition graph of $\D$ and $\D'$ the memoryless distribution of $G$
so that we have $\D \to G \to \D'$.
Since $\D$ and $\D'$ have the same transition graph we conclude that $\D'$ is
backwards $\alpha$-efficient.

For each $S \subset [n]$ we let $v_S \in S$ be an element of $S$ that maximizes
$w(S,S\setminus\set{v_S})$. Let $c$ be the cost function defined by
\begin{align}
	\notag
	c(v, S) = 
	\left \{
		\begin{array}{cc}
			\abs{S} & v = v_S \\
			0 & v \neq v_S
		\end{array}
	\right .
	\, .
\end{align}
We see that for each $S$ we have $\Epp{v \sim \U(S)}{c(v,S)}
= 1$. Since $\D'$ is backwards $\alpha$-efficient we have
that $\Epp{\pi \sim \D'}{c(\pi)} \le \alpha n$.
For each
  $i \in [n]$ let $t_i = \Epp{\pi \sim \D'}{c(\pi(i), \pi([i]))
  }$, then we have $\sum_{i = 1}^n t_i \le \alpha n$.

Let $q = \floor{\frac{n}{24\alpha}}$. If $q = 0$ then $\floor{\frac{n}{48\alpha}} = 0$
and there is nothing to prove, so assume that $q > 0$.
There are $\floor{\frac{n}{2q}}$ ways to choose a positive integer $k$
such that $n-kq \ge \frac{n}{2}$, and hence there must exist $p = n-kq \ge \frac{n}{2}$
such that $\sum_{i=p+1}^{p+q} t_i \le 
\frac{1}{\floor{\frac{n}{2q}}} \sum_{i=1}^n t_i \le 
\frac{\alpha n}{\floor{\frac{n}{2q}}}
\le 4\alpha q$.
Fix such a $p$.

For a permutation $\pi$ we now let $\ell(\pi)$ be the $(q+1)$-tuple defined by
$\ell(\pi) = (\pi([p]), \pi([p+1]), \ldots, \pi([p+q]))$.
We will bound the entropy of $\ell(\pi)$, when $\pi$ is chosen from the memoryless $\D'$,
i.e. $\H_{\pi \sim \D'}(\ell(\pi))$.
The $(q+1)$-tuple
$\ell(\pi)$ can be deduced from the two sets $\pi([p]), \pi([p+q])$ and the sequence
$(\pi(p+1),\ldots,\pi(p+q))$. The combined entropy of the two sets is 
at most $\H_{\pi \sim D'}(\pi([p]))+\H_{\pi \sim D'}(\pi([p+q]))$. The sequence 
$(\pi(p+1),\ldots,\pi(p+q))$ is a permutation of $\pi([p+q]) \setminus \pi([p])$,
so given $\pi([p])$ and $\pi([p+q])$ the entropy of the permutation is no
more than $\ln(q!)$. Therefore, the total entropy of $\ell(\pi)$ is at most 
\begin{align}
	\notag
	\H_{\pi \sim D'}(\pi([p]))+
	\H_{\pi \sim D'}(\pi([p+q]))+
	\ln(q!)
	\, .
\end{align}
The entropy $\H_{\pi \sim D'}(\pi([p]))$ can be calculated from the transition
graph $G$ of $\D'$. Since $\D$ and $\D'$ have the same transition graph it follows
that $\H_{\pi \sim D'}(\pi([p])) = \H_{\pi \sim D}(\pi([p]))$, and the latter is
bounded by $\H(\D)$. Using the same argument for $\H_{\pi \sim D'}(\pi([p+q]))$ 
gives that
\begin{align}
	\notag
	\H_{\pi \sim \D'}(\ell(\pi))
	\le 
	2\H(\D)
	+\ln(q!)
	\, .
\end{align}

Our goal is now to give a lower bound on the entropy of $\ell(\pi)$ when 
$\pi$ is chosen from $\D'$. We note that the upper bound on the entropy 
of $\ell(\pi)$ holds for
any distribution $\D_0$ such that $\D_0 \to G$, but for the lower bound we
will use that $\D'$ is the memoryless distribution of $G$. Since $\D'$ is the
memoryless distribution of $G$ we have that for any $(q+1)$-tuple of sets
$(X_q,X_{q+1},\ldots,X_{q+p})$ such that $\abs{X_i} = i$ and $X_i \subset X_{i+1}$
we have that
\begin{align}
	\notag
	& \phantom{{}={}}
	\Prpp{\pi' \sim D'}{
		\ell(\pi') = (X_q,\ldots,X_{q+p})
	}
	\\
	\notag
	& = 
	\Prpp{\pi' \sim D'}{ \pi'([q]) = X_q }
	\prod_{i=p+1}^{p+q}
		\Prpp{\pi' \sim D'}{
			\pi'([i]) = X_{i}
			\mid
			\pi'([i-1]) = X_{i-1}
		}
	\\
	\notag
	& =
	\Prpp{\pi' \sim D'}{\pi'([q]) = X_q}
	\prod_{i=p+1}^{p+q}
		w(X_i, X_{i-1})
	\\
	\notag
	& \le
	\prod_{i=p+1}^{p+q}
		w(X_i, X_{i-1})
	\, .
\end{align}
Inserting this into the definition of the entropy of $\ell(\pi)$ gives us that
\begin{align}
	\notag
	\H_{\pi \sim \D'}(\ell(\pi))
	\ge
	- \sum_{i=p+1}^{p+q} 
		\Epp{\pi \sim \D'}{
			\ln \left ( w(\pi([i]), \pi([i-1])) \right )
		}
	\, .
\end{align}
By the definition of $v_{\pi([i])}$ we have that $w(\pi([i]), \pi([i-1]))$ is
bounded by $w(\pi([i]), \pi([i]) \setminus v_{\pi([i])})$. Applying this
together with Jensen's inequality gives us that
\begin{align}
	\notag
	\H_{\pi \sim \D'}(\ell(\pi))
	& \ge
	- \sum_{i=p+1}^{p+q} 
		\Epp{\pi \sim \D'}{
			\ln \left ( w(\pi([i]), \pi([i]) \setminus v_{\pi([i])}) \right )
		}
	\\
	\notag
	& \ge
	- \sum_{i=p+1}^{p+q} 
		\ln \left ( 
			\Epp{\pi \sim \D'}{
				w(\pi([i]), \pi([i]) \setminus v_{\pi([i])})
			}
		\right )
	\, .
\end{align}
By the definition of $c$ we have that
\begin{align}
	\notag
	t_i = 
	\Epp{\pi \sim \D'}{ c(\pi(i), \pi([i])) }
	& =
	i \cdot
	\Epp{\pi \sim \D'}{
		w(\pi([i]), \pi([i]) \setminus v_{\pi([i])})
	}
	\\
	\notag
	& \ge 
	\frac{n}{2} \cdot
	\Epp{\pi \sim \D'}{
		w(\pi([i]), \pi([i]) \setminus v_{\pi([i])})
	}
	\, ,
\end{align}
and therefore the entropy of $\ell(\pi)$ is at least
\begin{align}
	\notag
	- \sum_{i=p+1}^{p+q} 
		\ln \left ( 
			\frac{t_i}{n/2}
		\right )
	=
	q \ln \left ( \frac{n}{2} \right )
	- \sum_{i=p+1}^{p+q} 
		\ln (t_i)
	\, .
\end{align}
By Jensen's inequality we get that 
\begin{align}
	\notag
	\sum_{i=p+1}^{p+q} 
		\ln (t_i)
	\le 
	q \ln \left (
		\frac{\sum_{i=p+1}^{p+q} t_i}{q}
	\right )
	\le 
	q \ln (4\alpha)
	\, .
\end{align}
Hence we have that $\H_{\pi \sim \D'}(\ell(\pi)) \ge q \ln \left ( \frac{n}{8\alpha} \right ) $.
Comparing this with the upper bound on the entropy gives us that
\begin{align}
	\notag
	2\H(\D) + \ln(q!) \ge
	q \ln \left ( \frac{n}{8\alpha} \right )
	\, .
\end{align}
Using that $\ln(q!) \le q\ln(q)$ and isolating $\H(\D)$ gives us that
\begin{align}
	\notag
	\H(\D) \ge 
	\frac{q}{2} \cdot
	\ln \left (
		\frac{n}{8\alpha q}
	\right )
	\ge 
	\frac{1}{2} \cdot
	\floor{\frac{n}{24\alpha}} \cdot
	\ln \left (
		3
	\right )
	\ge 
	\floor{\frac{n}{48\alpha}}
	\, ,
\end{align}
which was what we wanted.
\end{proof}

\section{Single Batch Lower Bound}

For a set $X$ we let $\binom{X}{k} \subset \P(X)$ be the set containing all subsets of $X$
size $k$. In this section we prove following result:
\begin{theorem}
	\label{thm:singleBatch}
	Let $k$ be an integer and $\D$ a distribution over $\binom{[n]}{k}$. Let $\alpha \ge 1$.
	Assume that for every cost-function $c$ such that $\Epp{s \sim \U(S')}{c(s,S')} \le 1$
	for any subset $S' \subset [n]$ of size $k$, it holds that:
	\begin{align}
		\label{eq:singleBatch}
		\Epp{S \sim \D, s \sim \U([n] \setminus S)}{c(s, S \cup \set{s})}
		\le 
		\alpha
		\, ,
	\end{align}
	If $\min\set{k,n-k} = \Omega(n)$
	then either
	$\alpha = \Omega(n)$ or $\H(\D) = \Omega \left ( \frac{n}{\alpha} \log \left ( 1 + \alpha \right ) \right )$.
\end{theorem}
Let $X$ be a set. We say that a function $w : \P(X) \to [0,1]$ is \emph{weight-function} for $X$. In order to
prove Theorem \ref{thm:singleBatch} we will first prove some properties of weight-functions.
We define $S(w)$ and $H(w)$ by:
\begin{align*}
	S(w) = \sum_{A \subset X} w(A)
	\, ,
	H(w) = \sum_{A \subset X} w(A) \ln \left ( \frac{1}{w(A)} \right )
	\, .
\end{align*}
In the definition of $H(w)$ we only sum over $A \subset X$ such that $w(A) \neq 0$.
We note that if $S(w) = 1$ then $w$ corresponds to a distribution $\D$ on the subsets
of $X$. In this case $H(w)$ is the entropy of $\D$.

For non-negative integers $i$ we define $\pi_i(w) = \pi_i w : \P(X) \to [0,1]$ by:
\begin{align*}
	(\pi_i w)(A) = \left [ w(A) \ge 2^{-i} \right ] 2^{-i-1}
	\, ,
\end{align*}
and we let $\Delta(w) = \Delta w : \P(X) \to [0,1]$ be defined by
$(\Delta w)(X) = 0$ and for $B \neq X$:
\begin{align*}
	(\Delta w)(B) = 
	\max_{a \notin B} \set{ w(B \cup \set{a}) }
	\, .
\end{align*}
We note that $\Delta \pi_i w = \pi_i \Delta w$. We first prove a lemma about two basic properties
of $\sum_{i \ge 0} S(\pi_i w)$ and $\sum_{i \ge 0} H(\pi_i w)$.
\begin{lemma}
	\label{lem:basicWeight}
	Let $w : \P(X) \to [0,1]$ be a weight-function, then it holds that:
	\begin{align}
		\label{eq:basicSumS}
		\sum_{i \ge 0} S(\pi_i w)
		& \in 
		\left [ \frac{1}{2} S(w), S(w) \right ]
		\\
		\label{eq:basicSumH}
		\sum_{i \ge 0} H(\pi_i w) & \le 
		3S(w) + 2H(w)
		\, .
	\end{align}
\end{lemma}
\begin{proof}
We first show \eqref{eq:basicSumS}.
Let $A$ be a subset of $X$ with $w(A) \neq 0$, and let $j$ be the smallest integer such that $w(A) \ge 2^{-j}$.
Then we have that:
\begin{align*}
	\sum_{i \ge 0} (\pi_i)(A) =
	\sum_{i \ge j} 2^{-j-1} =
	2^{-j}
	\in \left [ \frac{1}{2} w(A), w(A) \right ]
	\, ,
\end{align*}
and the conclusion follows.

We now show \eqref{eq:basicSumH}.
Let $A$ be a subset of $X$ with $w(A) \neq 0$ and let $j$ be the smallest integer such that
$w(A) \ge 2^{-j}$.
Then:
\begin{align*}
	\sum_{i \ge 0} (\pi_i w)(A) \ln \frac{1}{(\pi_i w)(A)} =
	\ln(2) \sum_{i \ge j} 2^{-i} \cdot i =
	\ln(2) \cdot 2^{-j+1} \cdot (j+1)
	\, .
\end{align*}
We clearly have that $2^{-j+1} \ge w(A)$ and therefore $j-1 \le \lg \frac{1}{w(A)}$, i.e.
$j+1 \le 2 + \lg \frac{1}{w(A)}$. Inserting this bound along with $2^{-j} \le w(A)$ gives that:
\begin{align*}
	\sum_{i \ge 0} (\pi_i w)(A) \ln \frac{1}{(\pi_i w)(A)}
	& \le
	2 w(A) \cdot \left (
		2\ln(2) + \ln \frac{1}{w(A)}
	\right )
	\\
	& =
	4\ln(2) \cdot w(A) + 2w(A) \ln \frac{1}{w(A)}
	\\
	& <
	3 w(A) + 2w(A) \ln \frac{1}{w(A)}
	\, .
\end{align*}
Using the definition of $H$ and $S$ gives
\begin{align*}
	\sum_{i \ge 0} H(\pi_i w) = 
	\sum_{i \ge 0, A \subset X} (\pi_i w)(A) \ln \frac{1}{(\pi_i w)(A)} \le
	\sum_{A \subset X} 3w(A) + 2w(A) \lg \frac{1}{w(A)}
	=
	3S(w) + 2H(w)
	\, ,
\end{align*}
as desired.
\end{proof}

For a real number $x$ and an non-negative integer $k$ we define 
$\binom{x}{k} = \frac{x(x-1)\cdots(x-k+1)}{k!}$.
We will use Lov\'{a}sz's version of the Kruskal-Katona Theorem, see e.g. \cite{frankl84kruskalkatona}.
\begin{theorem}[Lov\'{a}sz, Kruskal-Katona, \cite{frankl84kruskalkatona}]
	\label{thm:lovasz}
	Let $\A \subset \binom{X}{k}$, and $\abs{\A} = \binom{x}{k}$ for a real number $x \ge k$.
	Then,
	\begin{align*}
		\abs{\Delta \A} \ge \binom{x}{k-1} \, .
	\end{align*}
\end{theorem}
In order to make it easier to apply Theorem \ref{thm:lovasz} we define $\ell_k(x)$ in the following
way.
\begin{definition}
	\label{def:ellkx}
	For an integer $k$ and a non-negative real number $x$ let $\ell_k(x)$ be the unique non-negative
	real number such that $x = \binom{k-1+\ell_k(x)}{k}$.
\end{definition}
Using Definition \ref{def:ellkx} we can reformulate Theorem \ref{thm:lovasz} in the following manner:
\begin{corollary}
	\label{cor:lovaszNew}
	Let $\A \subset \binom{X}{k}$.
	Then,
	\begin{align*}
		\abs{\Delta \A} \ge \abs{\A} \cdot \frac{k}{\ell_k(\abs{\A})}  \, .
	\end{align*}
\end{corollary}
\begin{proof}
Let $x = \ell_k(\abs{\A})+k-1$. By definition we have that $\abs{\A} = \binom{x}{k}$ and so
we have that:
\begin{align*}
	\abs{\Delta \A} \ge \binom{x}{k-1}
	=
	\binom{x}{k} \cdot \frac{k}{x-k+1} =
	\abs{\A} \cdot \frac{k}{\ell_k(\abs{A})}
	\, ,
\end{align*}
as desired.
\end{proof}
Given $\A \subset \binom{X}{k}$ we can define a weight-function by $w_{\A}(A) = [A \in \A]$.
Corollary \ref{cor:lovaszNew} is now equivalent to
$S(\Delta w_{\A}) \ge S(w_{\A}) \cdot \frac{k}{\ell_k(S(w_{\A}))}$. However, it is not clear
what the correct generalisation of Corollary \ref{cor:lovaszNew} is to weight-functions with
non-uniform weights. Lemma \ref{lem:lovaszWeight} gives a bound $S(\Delta \pi_i w)$ in terms
of $S(\pi_i w)$.
\begin{lemma}
	\label{lem:lovaszWeight}
	Let $w : \P(X) \to [0,1]$ be a weight-function for $X$ with support on $\binom{X}{k}$.
	If $S(w) \le 1$ and $w(A) = 0$, then for every non-negative integer $i$:
	\begin{align*}
		S(\Delta \pi_i w)
		\ge 
		S(\pi_i w)
		\cdot
		\frac{k}{\ell_r(2^i)}
		\, .
	\end{align*}
\end{lemma}
\begin{proof}
Let $\A$ denote the support of $\pi_i w$. Clearly $\abs{\A} \le 2^i S(w) \le 2^i$ by the definition of $\pi_i$.
We note that $(\Delta \pi_i w)(B) = 2^{-i-1} \cdot [B \in \Delta\A]$.
By Theorem \ref{thm:lovasz} we have that $\abs{\Delta \A} \ge \abs{\A} \cdot \frac{k}{\ell_k(\abs{\A})}$.
We clearly have that
\begin{align*}
	S(\pi_i w) = \abs{\A} \cdot 2^{-i-1},
	\,
	S(\Delta \pi_i w) = \abs{\Delta \A} \cdot 2^{-i-1}
	\, .
\end{align*}
Inserting this gives
\begin{align*}
	S(\Delta \pi_i w)
	\ge 
	S(\pi_i w)
	\cdot
	\frac{k}{\ell_k(2^i)}
	\, ,
\end{align*}
as desired.
\end{proof}

\begin{lemma}
	\label{lem:singleBatchTechnical}
	Let $w : \P(X) \to [0,1]$ be a weight-function with $S(w) = 1$ and let $k$ be a
	non-negative integer the support of $w$ is contained in $\binom{X}{k}$.
	Assume that $S(\Delta w) \le \alpha$. Then either $\alpha = \Omega \left ( k \right )$ or
	$H(w) = \Omega \left ( \frac{k}{\alpha} \log \left(1+\alpha\right) \right )$.
\end{lemma}
\begin{proof}
Since $S(\Delta w) \le \alpha$ we conclude by Lemma \ref{lem:basicWeight} that
\begin{align*}
	\sum_{i \ge 0} S(\Delta \pi_i w)
	=
	\sum_{i \ge 0} S(\pi_i \Delta w)
	\le 
	S(\Delta w)
	\le \alpha
	\, .
\end{align*}
We also see that $\sum_{i \ge 0} S(\pi_i w) \ge \frac{1}{2}S(w) = \frac{1}{2}$. Therefore there exists a largest
integer $j$ such that $\sum_{i \ge j} S(\pi_i w) \ge \frac{1}{4}$. Since $j$ is the largest such integer we also have
that $\sum_{i=0}^j S(\pi_i w) \ge \frac{1}{4}$. Note that by Lemma \ref{lem:basicWeight}:
\begin{align*}
	2H(w) + 3 \ge 
	\sum_{i \ge 0} H(\pi_i w) = 
	\sum_{i \ge 0} S(\pi_i w) (i+1) \ge 
	(j+1) \sum_{i \ge j} S(\pi_i w) \ge
	\frac{1}{4} (j+1)
	\, ,
\end{align*}
and therefore we either have $j = O(1)$ or $H(w) = \Omega(j)$.

On the other hand Lemma \ref{lem:lovaszWeight} gives that
\begin{align*}
	\alpha \ge 
	\sum_{i=0}^j S(\Delta \pi_i w)
	\ge
	\sum_{i=0}^j S(\pi_i w) \frac{k}{\ell_k(2^i)}
	\ge 
	\frac{1}{4} \cdot \frac{k}{\ell_k(2^j)}
	\, ,
\end{align*}
i.e. $\ell_k(2^j) \ge \frac{k}{4\alpha}$. Let $\beta = \floor{\frac{k}{4\alpha}}-1$.
If $\beta < \frac{k}{8\alpha}$ or $j = O(1)$ then we have $\alpha = \Omega(k)$ as desired.
So assume that $\beta \ge \frac{k}{8\alpha}$ and $H(w) = \Omega(j)$. By definition of 
$\ell_k(2^j)$ we have that:
\begin{align*}
	2^j = 
	\binom{k-1 + \ell_k(2^j)}{k} \ge 
	\binom{k+\beta}{k} =
	\binom{k+\beta}{\beta} \ge 
	\left ( \frac{k+\beta}{\beta} \right )^\beta
	\, .
\end{align*}
Taking the logarithm on both sides yields:
\begin{align*}
	j \ge 
	\lg \left ( 1 + \frac{k}{\beta} \right )
	\cdot \beta
	\ge 
	\frac{k}{8\alpha} \cdot 
	\lg \left ( 1 + 4\alpha \right )
	\, .
\end{align*}
Since $H(w) = \Omega(j)$ we get $H(w) = \Omega \left ( \frac{k}{\alpha} \log (1+\alpha) \right )$
which was what we wanted.
\end{proof}

We now use Lemma \ref{lem:singleBatchTechnical} to prove Theorem \ref{thm:singleBatch}.
\begin{proof}[Proof of Theorem \ref{thm:singleBatch}]
Let $w : \P([n]) \to [0,1]$ be defined by $w(A) = \Prpp{S \sim \D}{S = [n] \setminus A}$. Then $w$ is a weigh-function
with $S(w) = 1$ and $H(w) = \H(\D)$. The support of $w$ is contained in $\binom{[n]}{n-k}$.

Let $f : \P(X) \to X$ be a mapping such that for every $B \in \binom{[n]}{k+1}$ we have
$f(B) \in B$ and
\begin{align*}
	w\left ( B \setminus \set{f(B)} \right )
	=
	\max_{b \in B} \set{w\left ( B \setminus \set{b} \right ) }
	\, ,
\end{align*}
that is, $b = f(B)$ maximizes $w \left ( B \setminus \set{b} \right )$ constrained to $b \in B$.
We define the cost function $c(b,B)$ for all $B \in \binom{[n]}{k+1}$ and $b \in B$ by
$c(b,B) = k+1$ if $b = f(B)$ and $c(b,B) = 0$ otherwise. By definition \eqref{eq:singleBatch}
must hold for this cost function. We plug $c$ into \eqref{eq:singleBatch} to obtain
\begin{align*}
	\frac{1}{n-k}
	\sum_{A \in \binom{[n]}{k}}
		\Prpp{S \sim \D}{S = A}
		\sum_{a \notin A}
			c(a, A \cup \set{a})	
	\le
	\Epp{S \sim \D, s \sim \U([n] \setminus S)}{c(s, S \cup \set{s})}
	\le 
	\alpha
	\, .
\end{align*}
We note that by the definition of $f(B)$ we get:
\begin{align*}
	\frac{1}{n-k}
	\sum_{A \in \binom{[n]}{k}}
		\Prpp{S \sim \D}{S = A}
		\sum_{a \notin A}
			c(a, A \cup \set{a})
	& =
	\frac{1}{n-k}
	\sum_{B \in \binom{[n]}{k+1}}
		\sum_{b \in B}
		\Prpp{S \sim \D}{S = B \setminus \set{b}}
		c(b, B)
	\\
	& =
	\frac{k+1}{n-k}
	\sum_{B \in \binom{[n]}{k+1}}
		\Prpp{S \sim \D}{S = B \setminus \set{f(B)}}
\end{align*}
By the definition of $\Delta w$ we get that:
\begin{align*}
	\Prpp{S \sim \D}{S = B \setminus \set{f(B)}}
	=
	\max_{b \in B} \set {
		\Prpp{S \sim \D}{S = B \setminus \set{b}}
	}
	=
	\max_{b \notin [n] \setminus B} \set{ w(([n]\setminus B) \cup \set{b}) }
	=
	(\Delta w)([n] \setminus B)
	\, .
\end{align*}
That is,
\begin{align*}
	\alpha
	\ge 
	\frac{k+1}{n-k}
	\sum_{B \in \binom{[n]}{k+1}}
		\Prpp{S \sim \D}{S = B \setminus \set{f(B)}}
	=
	\frac{k+1}{n-k}
	\sum_{B' \in \binom{[n]}{n-k-1}}
		(\Delta w)(B')
	=
	\frac{k+1}{n-k}
	S(\Delta w)
	\, .
\end{align*}
Since $k = \Omega(n)$ we get that $S(\Delta w) = O(\alpha)$.
Now Lemma \ref{lem:singleBatchTechnical} gives that either $\alpha = \Omega(n-k)$ or
$\H(\D) = H(w) = \Omega \left ( \frac{n-k}{\alpha} \log (1 + \alpha) \right )$. Since $n-k = \Omega(n)$
this is exactly what we wanted to prove.
\end{proof}

\section{Upper bound}


We present a general lemma about transition graphs that will help us find
distributions that are uniform on relatively small sets that have 
desirable properties. The lemma itself is fairly straight-forward and
similar statements might be known.

\begin{lemma}
	\label{lem:pebbleDistribution}
	Let $G$ be a transition graph, and let $\delta_V$ and $\delta_E$ be the
	minimal non-zero weight of a node and an edge in $G$, respectively,
	and assume that $w_G(S,S')^{-1}$ is an integer for every edge $(S,S')$
	with non-zero weight.
	Let $p$ be the probability of the most probable permutation in the
	memoryless distribution of $G$.
	Let $\eps \in (0,1]$, and $t$ be integer satisfying
	$\frac{1}{p} \ge t \ge \frac{8n}{\delta_V \delta_E \eps}$.
	There exists a set $X$ such that the transition graph $G'$ of $\U(X)$
	satisfies:
	For each node $S$, $w_{G'}(S) = (1 \pm \eps)w_G(S)$ and for each edge
	$(S,S')$, $w_{G'}(S,S') = \left(1 \pm \frac{\eps}{4n}\right)w_G(S,S')$.
	
	Furthermore, if $t \cdot w_G(S) \cdot w_G(S,S')$ is an integer for 
	each edge $(S,S')$ of $G$ then $G' = G$.
\end{lemma}
\begin{proof}
Put $t$ pebbles on the node $[n]$ of $G$ and $0$ pebbles on each other node of
$G$. We now run the following algorithm:
For each $i = n,n-1,n-2,\ldots,1$ do the following. For each subset $S \subset [n]$
of size $i$ consider the node $S$ of $G$, and say that there are currently $s$ pebbles
on $S$. Now let $T_1,\ldots,T_i$ be all subsets of $S$ of size $i-1$. We then move
all the pebbles from $S$ to the nodes $T_1,\ldots,T_i$ such that we move
$\approx s w_G(S,T_i)$ pebbles from $S$ to $T_i$. That is, we move at least
$\floor{s w_G(S,T_i)}$ pebbles and at most $\ceil{s w_G(S,T_i)}$ pebbles from $S$
to $T_i$. We do this in the following way. Let $V$ be the set of pebbles in $S$
and partition $V$ as $V = V_1 \cup \ldots V_k$ where each $V_i$ holds a set
of pebbles that have all followed the same path from $[n]$ to $S$. Then we
use Lemma \ref{lem:evenSplit} with $\delta_i = w_G(S,T_i)$ to create a function
$f : V \to [k]$. A pebble $v \in V$ is then moved to $T_{f(v)}$.

When the algorithm is finished all the pebbles are at the node $\emptyset$. 
Each pebble went from $[n]$ to $\emptyset$ through a path $[n] = S_n, S_{n-1}, \ldots, S_0 = \emptyset$.
The pebble then correspond to the unique permutation satisfying that 
$\set{\pi(i)} = S_i \setminus S_{i-1}$.
For any permutation $\pi$, since we used Lemma \ref{lem:evenSplit} to move the
pebbles, the number of pebbles that took the path corresponding to $\pi$,
say $S_n, S_{n-1}, \ldots, S_0$, is at most
\[ \ceil{w_G(S_1,S_0)\ceil{\ldots\ceil{w_G(S_n,S_{n-1})t}\ldots}}
= \ceil{w_G(S_1,S_0) \cdots w_G(S_n,S_{n-1})\cdot t}, \]
which equals $\ceil{p't}$, where $p'$ is the probability of getting $\pi$
when drawing a permutation from the memoryless distribution of $G$.
Since $p't \le 1$ by assumption this shows that any two pebbles followed
different paths in $G$ and therefore correspond to different permutations.
Let $\D$ be the uniform distribution
of the permutations corresponding to the pebbles, and let $G'$ be the 
transition graph of $\D$.

We will prove that $w_{G'}(S,S') = \left ( 1 \pm \frac{\eps}{4n} \right )w_G(S,S')$
for all edges $(S,S')$. Assume for the sake of contradiction that this is false, and
let $(S,S')$ be an edge for which it does not hold with the largest possible value
of $\abs{S}$.
We can calculate the weight of $S$ in $G'$ by considering all paths from $[n]$ to
$S$ and summing the products of the weights of the edges on each path. These edge
weights are at most a factor $1-\frac{\eps}{4n}$ smaller than the corresponding 
edge weights in $G$, which can be used to calculate $w_G(S)$. This implies
that the weight of $S$ in $G'$ is at least 
$\left ( 1 - \frac{\eps}{4n} \right )^{n-\abs{S}}w_G(S) \ge 
\left ( 1 - \frac{\eps}{4n} \right )^n w_G(S)$, and hence $w_{G'}(S)$ is at least
$\left(1-\eps/2\right)w_G(S) \ge \frac{1}{2} w_G(S) \ge \frac{1}{2}\delta_V$.
The number of pebbles that where in $S$ was exactly $t w_{G'}(S)$.
So the number of pebbles that where moved from $S$
to $S'$ is either $\floor{t w_{G'}(S)w_G(S,S')}$ or $\ceil{t w_{G'}(S)w_G(S,S')}$.
The weight $w_{G'}(S,S')$ is therefore at most
\begin{align}
	\notag
	\frac{\ceil{t w_{G'}(S)w_G(S,S')}}{t w_{G'}(S)}
	< 
	\frac{1 + t w_{G'}(S)w_G(S,S')}{t w_{G'}(S)}
	& =
	w_G(S,S') + \frac{1}{t w_{G'}(S)}
	\\
	\notag
	& \le 
	\left ( 1 + \frac{\eps}{4n} \right )
	w_G(S,S')
	\, ,
\end{align}
where the last inequality follows from the fact that 
$\frac{1}{t w_{G'}(S)} \le \frac{2}{t w_G(S)}$ and the definition of $t$.
We get that $w_{G'}(S,S')$ is at least $\left ( 1 - \frac{\eps}{4n} \right )w_G(S,S')$
in the same manner. A contradiction. So the assumption was wrong and 
$w_{G'}(S,S') = \left ( 1 \pm \frac{\eps}{4n} \right )w_G(S,S')$ holds for every
edge of $G'$.

By using the fact that we can calculate the node weights from the edge weights
it is apparent that $w_{G'}(S) = \left ( 1 \pm \frac{\eps}{4n} \right )^nw_G(S)$
which clearly implies $w_{G'}(S) = \left ( 1 \pm \eps \right ) w_G(S)$ as desired.
\end{proof}

Let $G$ be the transition graph of $\U(S_n)$ and let 
$t = \lcm_{k \in [n]} \set { \binom{n}{k} \cdot k }$. For each edge $(S,S')$ 
of $G$ we have $w_G(S)w_G(S,S') = \binom{n}{k}^{-1} \cdot k^{-1}$ and $tw_G(S) w_G(S,S')$
is an integer. Therefore $G' = G$ in Lemma \ref{lem:pebbleDistribution}, and there exists
a distribution, which is uniform on a set of $t$ permutations that has the same transition
graph as $\U(S_n)$. By Proposition \ref{prop:minwiseEqualsBackwards} this distribution is
exact minwise and exact backwards uniform as well. It is an easy exercise to prove
that $t = \lcm(1,2,\ldots,n)$, and the proof is deferred to Lemma \ref{lem:numberTheoryLCM}
The result is summarized below.

\begin{corollary}
	\label{cor:optimalMinWiseUniform}
	There exists a subset $X \subset S_n$ of size
	$\lcm(1,2,\ldots,n) = e^{n-o(n)}$ such that $\U(X)$ is
	exact minwise and exact backwards uniform.
\end{corollary}
This is tight, since it is shown in \cite{broder98minwise} that for any uniform
distribution on $X$ that is exact minwise, the size of $X$ is
divisible by $\lcm(1,2,\ldots,n)$. This also improves on the previous upper bound
from \cite{broder98minwise} of $4^{n-O(\log^2 n)}$.

Now we will show that for any $\alpha > 1$ there exists a backwards 
$\alpha$-uniform distribution
$\D$ that is the uniform distribution on
$\exp\left(O\left(\frac{n}{\alpha}\cdot \left( 1 + \log^2 \alpha \right )\right)\right )$
permutations. We will do this by first describing a preliminary distribution $\D'$, and
then convert $\D'$ into a similar distribution $\D$ using Lemma \ref{lem:pebbleDistribution}.

For simplicity we assume that $n$ and $\alpha$ are powers of two. Let $t = \frac{n}{\alpha}$.
We now describe an algorithm that takes a positive integer $k$ and a set $S$ of
$2^k t$ integers and outputs a permutation of $S$. We describe the permutation, that
the algorithm finds when $S = [2^kt]$. For general $S$ we order the elements of 
$S$ as $S = \set{s_1,s_2,\ldots,s_{2^kt}}$ such that $s_1 < s_2 < \ldots < s_{2^kt}$.
Then we find a random permutation $\pi'$ of $[2^kt]$, and the algorithm outputs
$\pi : S \to S$ defined by $\pi(s_i) = s_{\pi(i)}$. So from now on we may
assume that $S = [2^k t]$. If $k = 1$ the algorithm
simply chooses a permutation uniformly at random $S$ and outputs that permutation. If
$k > 1$ the algorithm partitions $S$ into two sets $S_0 = [2^{k-1}t]$ and
$S_1 = [2^kt] \setminus [2^{k-1}t]$ consisting of
the smallest and largest half of $S$, respectively. The algorithm
now recursively finds permutations of $S_0$ and $S_1$, which we call $\pi_0$
and $\pi_1$, respectively. We assume that $\pi_0$ and $\pi_1$ are independent.
Let $S'$ be the set of the $t$ last elements of $\pi_0$ and $\pi_1$, i.e.,
\begin{align}
	\notag
	S' = \pi_0 \left ( S_0 \setminus [2^{k-1}t-t] \right )
		\cup 
		\pi_1 \left ( S_1 \setminus [2^t-t] \right )
	\, .
\end{align}
Let $s'_1,\ldots,s'_{2t}$ be a uniformly random ordering of $S'$. The result of
the algorithm is now obtained by concatenating the first $2^{k-1}t-t$ elements
of $\pi_0$ with the first $2^{k-1}t-t$ elements of $\pi_1$ and then appending 
$s'_1, \ldots, s'_{2t}$. That is, the output of the algorithm is the following
permutation:
\begin{align}
	\notag
	\pi(i) 
	=
	\left \{
		\begin{array}{cc}
			\pi_0(i) & i \le 2^{k-1}t-t \\
			\pi_1 \left ( i - (2^{k-1}t-t) \right ) & 2^{k-1}t-t < i \le 2^kt-2t \\
			s'_{i-(2^k-2t)} & 2^k-2t < i 
		\end{array}
	\right .
	\, .
\end{align}
For each value of $k$ this gives a distribution of permutations on $\set{1,2,\ldots,2^k t}$
which we call $\D_k$.

We now prove Lemma \ref{lem:UBLastElements} and Lemma \ref{lem:UBFirstElements} which
give a lower bound on the non-zero weights of the nodes in the transition
graph of $\D_k$.

\begin{lemma}
	\label{lem:UBLastElements}
	Let $k$ be a positive integer and let $X \subset [2^kt]$ be a set of
	size $s$. The probability that all elements of $X$ are among the last $t$ elements
	in a permutation chosen randomly according to $\D_k$ is at least $2^{-2sk}$.
\end{lemma}
\begin{proof}
We prove the claim by induction on $k$. First say that $k=1$. Then the probability
that all elements from $X$ are among the last $t$ elements is exactly
$\frac{\binom{t}{s}}{\binom{2t}{s}} > 2^{-2s}$ as desired.

Now say that $k > 1$, and assume that the claim holds for the distribution
$\D_{k-1}$.
Let $A_0 = [2^{k-1}t], A_1 = [2^kt] \setminus [2^{k-1}t]$
and set $X_0 = X \cap A_0, X_1 = X \cap A_0$. The algorithm first computes 
two permutations of $A_0$ and $A_1$. The probability that all elements of $X_i$
are among the last $t$ elements in the permutation of $A_i$ is at least 
$2^{-2(k-1)\abs{X_i}}$ by the induction hypothesis. So by the independence of
the two permutations, the probability that all elements of $X_i$ are among the last $t$
elements in the permutation of $A_i$ for $i=0$ and $i=1$ is at least
$2^{-2(k-1)\abs{X_0}} \cdot 2^{-2(k-1)\abs{X_1}} = 2^{-2(k-1)s}$. That
is, the probability that all elements of $X$ are among the last $2t$ elements
of the permutation is at least $2^{-2(k-1)s}$. Now assume that this is the case.
Then using the same analysis as in the case $k=1$ we see that that the probability that
all elements of $X$ are among the last $t$ elements in the permutation is at least 
$\frac{\binom{t}{s}}{\binom{2t}{s}} > 2^{-2s}$. So the probability that all
elements of $X$ are contained among the last $t$ elements in the permutation
must be at least $2^{-2(k-1)s} \cdot 2^{-2s} = 2^{-2ks}$ as desired.
\end{proof}

\begin{lemma}
	\label{lem:UBFirstElements}
	Let $k$ be a positive integer and $X \subset [2^kt]$ a set of
	$s$ positive integers. The probability the the set of the first $s$ elements
	is $X$ in a permutation chosen randomly according to $\D_k$ is either $0$
	or at least $2^{-2k^2t}$.
\end{lemma}
\begin{proof}
We prove the claim by induction on $k$. First assume that $k=1$. Then the probability
is exactly $\frac{1}{\binom{2t}{s}} > 2^{-2t} = 2^{-2k^2t}$ as desired.

Now say that $k > 1$ and assume that the claim holds for $\D_{k-1}$.
Let $A_0 = [2^{k-1}t], A_1 = [2^kt] \setminus [2^{k-1}t]$
and set $X_0 = X \cap A_0, X_1 = X \cap A_0$. The algorithm first recursively
computes  permutations $\pi_i$ of $A_i$ for $i=0,1$ and then uses them to compute
a permutation $\pi$ of $[2^kt]$.
We now consider three cases:
\begin{itemize}
	\item[(1)] $s \le (2^{k-1}-1)t$.
	\item[(2)] $(2^{k-1}-1)t < s \le (2^k-2)t$.
	\item[(3)] $(2^k-2)t < s$.
\end{itemize}

\textbf{Case (1)}: Whether the first $s$ elements of $\pi$
is $X$ depends only on $\pi_0$, and by the induction hypothesis
the probability is either $0$ or at least $2^{-2(k-1)^2t} > 2^{-2k^2t}$.

\textbf{Case (2)}: The first $(2^{k-1}-1)t$ elements
of $\pi$ are elements from $A_0$ and the next $s-(2^{k-1}-1)t$ elements
are elements from $A_1$. This means that if $\abs{X_0} \neq (2^{k-1}-1)t$ the
probability is $0$, so assume that $\abs{X_0} = (2^{k-1}-1)t$. Now the first
$s$ elements of $\pi$ is $X$ if and only if the last $t$ elements 
of the $\pi_0$ is $A_0 \setminus X_0$ and the first 
$s-(2^{k-1}-1)t$ elements of the permutation of $A_1$ is $X_1$. By
Lemma \ref{lem:UBLastElements} and the induction hypothesis applied to $\D_{k-1}$
this happens with probability at least
$2^{-2(k-1)t} \cdot 2^{-2(k-1)^2t} > 2^{-2k^2t}$.

\textbf{Case (3)}: The set of the first $s$ elements of $\pi$ is 
$X$ exactly if the last elements of $\pi$ are $[2^kt] \setminus X =
A_0 \setminus X_0 \cup A_1 \setminus X_1$.
If $A_i \setminus X_i$ contains more than $t$ elements for $i=0$ or $i=1$
the probability that this happens is $0$, so assume that $\abs{A_i \setminus X_i} \le t$.
The last elements of $\pi$ are $A_0 \setminus X_0 \cup A_1 \setminus X_1$
only if the last $t$ elements of $\pi_i$ contains $A_i \setminus X_i$
for $i=0,1$. The latter happens with probability at least
$\left ( 2^{-2(k-1)t} \right )^2$ by Lemma \ref{lem:UBLastElements}. If this happens
then the probability that the last of $\pi$ is $A_0 \setminus X_0 \cup A_1 \setminus X_1$
$\binom{2t}{\abs{A_0 \setminus X_0 \cup A_1 \setminus X_1}}^{-1} > 2^{-2t}$. So
the probability that the first $s$ elements of $\pi$ is $X$ is either $0$
or at least $\left ( 2^{-2(k-1)t} \right )^2 \cdot 2^{-2t} > 2^{-2k^2t}$.

So in each case the probability is either $0$ or at least $2^{-2k^2t}$. By induction
this holds for all positive integers $k$.
\end{proof}

Lemma \ref{lem:UBFirstElements} shows that the weight of any node in the 
transition graph of $\D_{\lg \alpha}$ is either $0$ or
at least
$\exp \left (-O\left ( \frac{n}{\alpha} \cdot \left ( 1 + \log^2 \alpha \right ) \right ) \right )$.
Combining this with Lemma \ref{lem:pebbleDistribution} shows that there exists a
set $X$ of size 
$\exp \left (O\left ( \frac{n}{\alpha} \cdot \left ( 1 + \log^2 \alpha \right ) \right ) \right )$
such that $\U(X)$ is backwards $\alpha$-uniform.

\newpage
\bibliographystyle{plain}
\bibliography{general}

\newpage
\section{Appendix}

\begin{lemma}
	\label{lem:numberTheoryLCM}
	For any integer $n$ we have that:
	\begin{align}
		\label{eq:numberTheoryLCM}
		\lcm_{k \in [n]} \left (
			\binom{n}{k} \cdot k
		\right )
		=
		\lcm(1,2,\ldots,n)
		\, .
	\end{align}
\end{lemma}
\begin{proof}
Fix $n$ and let $a = \lcm_{k \in [n]} \left ( \binom{n}{k} k \right )$ and 
$b = \lcm(1,2,\ldots,n)$.
For a prime $p$ and a integer $m$ we let $v_p(m)$ denote the largest integer
$\ell$ such $p^\ell$ divides $m$. It is well-known that $v_p(m!)$ can be
calculated by the following sum:
\begin{align}
	\label{eq:vpForFactorials}
	v_p(m!)
	=
	\sum_{i \ge 1}
		\floor{\frac{m}{p^i}}
	\, .
\end{align}
Using that $\binom{n}{k} = \frac{n!}{(n-k)!(k-1)!}$ and \eqref{eq:vpForFactorials}
we get that
\begin{align}
	\label{eq:expressionForvpa}
	v_p(a) = 
	\max_{k \in [n]} \set {
		\sum_{i \ge 1} \left (
			\floor{\frac{n}{p^i}}
			-
			\floor{\frac{n-k}{p^i}}
			-
			\floor{\frac{k-1}{p^i}}
		\right )
	}
	\, .
\end{align}
We will prove that $v_p(a) = v_p(b)$ for every prime $p$. This clearly implies that $a=b$.
Fix a prime $p$, and let $\ell = v_p(b)$. Then $p^\ell \le n$ and $p^{\ell+1} \ge n$.
Therefore we can write $v_p(a)$ as:
\begin{align}
	\notag
	v_p(a) & = 
	\max_{k \in [n]} \set {
		\sum_{i = 1}^\ell \left (
			\floor{\frac{n}{p^i}}
			-
			\floor{\frac{n-k}{p^i}}
			-
			\floor{\frac{k-1}{p^i}}
		\right )
	}
	\\
	\notag
	\phantom{{}v_p(a){}}
	&
	\le 
	\max_{k \in [n]} \set {
		\sum_{i = 1}^\ell \left (
			\left ( \frac{n}{p^i} + \frac{p^i-1}{p^i} \right )
			-
			\frac{n-k}{p^i}
			-
			\frac{k-1}{p^i}
		\right )
	}
	=
	\ell
	\, .
\end{align}
To prove that $v_p(a) \ge \ell$ we let $k = p^\ell$ in \eqref{eq:expressionForvpa}
to obtain that
\begin{align}
	\notag
	v_p(a) & \ge 
	\sum_{i = 1}^\ell \left (
		\floor{\frac{n}{p^i}}
		-
		\floor{\frac{n-p^\ell}{p^i}}
		-
		\floor{\frac{p^\ell-1}{p^i}}
	\right )
	\\
	\notag
	\phantom{{}v_p(a){}}
	&
	=
	\sum_{i = 1}^\ell \left (
		\floor{\frac{n}{p^i}}
		-
		\left (
			\floor{\frac{n}{p^i}}
			-p^{\ell-i}
		\right )
		-
		\left ( 
			p^{\ell-i}-1
		\right )
	\right )
	=
	\ell
	\, .
\end{align}
Hence $v_p(a) = \ell = v_p(b)$, and since this holds for every prime
$p$ we have $a=b$ and \eqref{eq:numberTheoryLCM} holds as desired.
\end{proof}

\begin{lemma}
	\label{lem:evenSplit}
	Let $V = V_1 \cup V_2 \cup \ldots V_k$ be a partition of a set $V$.
	Let $r$ be an integer and let $\delta_1,\delta_2,\ldots,\delta_r \ge 0$ be non-negative
	real numbers such that $\sum_{i=1}^r \delta_i = 1$. Then there exists
	a function $g : V \to [r]$ such that for all $i \in [r]$ and $j \in [k]$ we
	have
	\begin{align}
		\notag
		\floor{\delta_i \abs{V_j}} \le  
		\abs{g^{-1}(i) \cap V_j} \le 
		\ceil{\delta_i \abs{V_j}}
		\, ,
	\end{align}
	and such that it also holds that
	\begin{align}
		\notag
		\floor{\delta_i \abs{V}} \le  
		\abs{g^{-1}(i) \cap V} \le 
		\ceil{\delta_i \abs{V}}
	\end{align}
\end{lemma}
\begin{proof}
Let $G$ be a with the nodes $U = \set{s,t} \cup \set{V_1,\ldots,V_k} \cup [r]$, where
$s$ and $t$ are a source and a sink, respectively. For each $j \in [k]$ let
$\alpha_j = \abs{V_j} - \sum_{i \in [r]} \floor{\delta_i \abs{V_j}}$, and let
$\beta_i = \ceil{\abs{V}\delta_i} - \sum_{j \in [k]} \floor{\delta_i \abs{V_j}}$.
The graph has an edge from $s$ to $V_j$ with capacity $\alpha_j$ for every $j \in [k]$,
and an edge from $i$ to $t$ with capacity $\beta_i$ for every $i \in [r]$.
For each pair $(V_j,i)$ such that $\delta_i \abs{V_j}$ is not an integer we add
an edge from $V_j$ to $i$ with capacity $1$.

We will now show that the maximum $st$-flow in $G$ is $\alpha = \sum_{j \in [k]} \alpha_j$. Since
the cut $(s, U \setminus \set{s})$ has value $\alpha$, the flow is at most $\alpha$. On
the other hand let $f$ be a flow function defined as follows. For any $j \in [k]$
$f(s,V_j) = \alpha_j$. For any $j \in [k], i \in [r]$
$f(V_j, i) = \delta_i \abs{V_j} - \floor{\delta_i \abs{V_j}}$. For any $i \in [r]$
$f(i, t) = \abs{V}\delta_i - \sum_{j \in [k]} \floor{\delta_i \abs{V_j}}$.
It is easy to verify that $f$ is a flow with value $\alpha$.

Since all the capacities of $G$ are integers, there exists a maximum $st$-flow $f_I$
where the flow along each edge is an integer. We now construct $g$ from $f_I$.
For each $j \in [k]$ we partition $V_j$ into $V_j = V_j^1 \cup V_j^2 \cup \ldots \cup V_j^r$
such that $\abs{V_j^i} = \floor{\delta_i \abs{V_j}} + f_I(V_j,i)$ for every $i \in [r]$.
This is clearly possible. Now we define $g$ such that $g(v) = i$ for every $v \in \bigcup_{j \in [k]} V_j^i$,
and it is easy to see that $g$.
\end{proof}

\end{document}